\documentclass[12pt,3p]{nj_style}

\usepackage{amsmath}
\usepackage{amssymb}
\usepackage{amsthm}
\usepackage{enumerate}
\usepackage[font=small,labelfont=bf,margin=8pt]{caption}
\usepackage[usenames,dvipsnames]{color}
\usepackage{graphicx}
\usepackage[colorlinks=true]{hyperref}

\newtheorem{thm}{Theorem} 

\newtheorem{prop}[thm]{Proposition}

\newcommand{\bra}[1]{\langle #1 |}
\newcommand{\ket}[1]{| #1 \rangle}

\newcommand{\ketbra}[2]{| #1 \rangle\langle #2 |}
\newcommand{\Tr}{{\rm Tr}}
\newcommand{\bb}[1]{\mathbb{#1}}
\newcommand{\cl}[1]{\mathcal{#1}}

\newcommand{\mnorm}[1]{%
\left\vert\kern-0.9pt\left\vert\kern-0.9pt\left\vert #1
\right\vert\kern-0.9pt\right\vert\kern-0.9pt\right\vert}
\newcommand{\bmnorm}[1]{%
\big\vert\kern-0.9pt\big\vert\kern-0.9pt\big\vert #1
\big\vert\kern-0.9pt\big\vert\kern-0.9pt\big\vert}

\begin{document}

\begin{frontmatter}



\title{Duality of Entanglement Norms}

\author[UG,IQC]{Nathaniel Johnston}
\ead{nathaniel.johnston@uwaterloo.ca}
\author[UG,IQC]{David W. Kribs}
\ead{dkribs@uoguelph.ca}

\address[UG]{Department of Mathematics \& Statistics, University of Guelph, Guelph, Ontario N1G~2W1, Canada}
\address[IQC]{Institute for Quantum Computing, University of Waterloo, Waterloo, Ontario N2L~3G1, Canada}

\begin{abstract}
	We consider four norms on tensor product spaces that have appeared in quantum information theory and demonstrate duality relationships between them. We show that the product numerical radius is dual to the robustness of entanglement, and we similarly show that the $S(k)$-norm is dual to the projective tensor norm. We show that, analogous to how the product numerical radius and the $S(k)$-norm characterize $k$-block positivity of operators, there is a natural version of the projective tensor norm that characterizes Schmidt number. In this way we obtain an elementary new proof of the cross norm criterion for separability, and we also generalize both the cross norm and realignment criteria to the case of arbitrary Schmidt number.
\end{abstract}


\end{frontmatter}

\section{Introduction}

In quantum information theory, the Schmidt rank and Schmidt number are some of the most basic measures of entanglement. The Schmidt rank gives a rough measure of the amount of entanglement contained in a pure state, and Schmidt number is the natural generalization that describes the amount of entanglement within an arbitrary (i.e., potentially mixed) quantum state \cite{TH00}.

Several norms and measures based on the Schmidt rank and Schmidt number have been defined. In this paper, we focus in particular on four such norms: the \emph{product numerical radius} \cite{GPMSZ10,PGMSCZ11}, the \emph{$S(k)$-norm} \cite{JK10,JK11}, the \emph{projective tensor norm} \cite{R00,Rud05}, and the \emph{robustness of entanglement} \cite{VT99}. We show that these four norms are all closely-related in many ways. Most notably, the product numerical radius is dual to the robustness of entanglement and the $S(k)$-norm is dual to the projective tensor norm.

These duality relationships allow us to construct new bounds on these norms, and we are able to generalize and provide succinct new proofs of many known results as well. We also get several results for free as a result of this duality, such as the structure of the isometry group of the projective tensor norm.

The remainder of this paper is organized as follows. In Section~\ref{sec:dual_norms}, we present the mathematical background that is relevant for our work. In particular, we introduce dual norms, discuss several special cases in which the characterization of a norm's dual is known, and present our main result, which gives an alternate characterization of dual norms in general. In Section~\ref{sec:quantum}, we introduce the basics of quantum information theory and quantum entanglement. We also present and motivate the four entanglement norms that we consider throughout the remainder of the paper, and discuss their basic properties. We then apply our general duality result to those entanglement norms in Section~\ref{sec:dual_ent_norms}, and use their duality to present an elementary proof of a generalization of the cross norm criterion for entanglement.

The remaining sections are devoted to presenting other applications of the duality of the entanglement norms. In Section~\ref{sec:pure_states} we use duality to compute the value of these norms on pure states (i.e., rank-$1$ operators). In Section~\ref{sec:isometry}, we use duality to determine the structure of the isometry group of one of the norms under consideration. Finally, in Section~\ref{sec:realignment}, we use duality to generalize the realignment criterion \cite{CW03,R03} and the filter covariance matrix criterion \cite{GGHE08} to arbitrary Schmidt number.

\section{Norms and Dual Norms}\label{sec:dual_norms}

We use $\cl{H}$ to denote a finite-dimensional Hilbert space over the field $\bb{F}$ of real or complex numbers ($\bb{R}$ or $\bb{C}$, respectively). We typically represent vectors $\mathbf{v} \in \cl{H}$ using boldface, but if we wish to emphasize that the vector in question has unit length (with respect to the norm induced by the inner product), then we use ``ket'' notation: $\ket{v} \in \cl{H}$. In this case, we use ``bras'' to represent dual (i.e., row) vectors: $\bra{v} := \ket{v}^\dagger$.

Given a norm $\mnorm{\cdot}$ on $\cl{H}$ (not necessarily equal to the norm induced by the inner product), the \emph{dual norm} of $\mnorm{\cdot}$ is defined by
\begin{align}\label{eq:dual_defn}
	\mnorm{\mathbf{v}}^\circ := \sup\Big\{ \big| \langle \mathbf{w}, \mathbf{v} \rangle \big| : \mnorm{\mathbf{w}} \leq 1 \Big\}.
\end{align}

For example, the dual of the vector $p$-norm is the vector $q$-norm, where $1/p + 1/q = 1$. To see some other well-known and useful duality relations, let $\cl{H} = M_n$, the space of $n \times n$ complex matrices with the Hilbert--Schmidt inner product $\langle A, B \rangle := \Tr(A^\dagger B)$ (here $A^\dagger$ denotes the conjugate transpose of $A$). The \emph{Frobenius norm} is the norm induced by this inner product and is thus self-dual: $\|X\|_F := \sqrt{\Tr(X^\dagger X)}$. The \emph{operator norm} and the \emph{trace norm} on $M_n$ are defined as follows:
\begin{align*}
	\big\|X\big\| & := \sup\Big\{ \big|\bra{v}X\ket{w}\big| \Big\} \ \text{ and } \	\big\|X\big\|_{tr} := \sup\Big\{ \big|\Tr(XU)\big| : \text{ $U \in M_n$ is unitary} \Big\}.
\end{align*}
It is well-known that the operator norm and the trace norm are dual to each other: $\|\cdot\|^\circ = \|\cdot\|_{tr}$. More generally, we have the following useful result, which follows immediately from \cite[Theorem~3.3]{MF85}.
\begin{thm}\label{thm:2k_dual}
	Let $X \in M_{n,m}$ have singular values $\sigma_1 \geq \sigma_2 \geq \cdots \geq \sigma_{\min\{m,n\}} \geq 0$ and define the $(k,2)$-norm of $X$ as follows:
	\begin{align*}
		\big\|X\big\|_{(k,2)} := \sup\Big\{ \big| \Tr(XY) \big| : {\rm rank}(Y) \leq k, \big\|Y\big\|_F \leq 1 \Big\} = \sqrt{\sum_{i=1}^k \sigma_i^2}.
	\end{align*}
	Let $r$ be the largest index $1 \leq r < k$ such that $\sigma_r > \sum_{i=r+1}^{\min\{m,n\}}\sigma_{i}/(k-r)$ (take $r = 0$ if no such index exists or if $k = 1$). Also define $\tilde{\sigma} := \sum_{i=r+1}^{\min\{m,n\}}\sigma_i/(k-r)$. Then	
	\begin{align*}
		\big\|X\big\|_{(k,2)}^\circ = \sqrt{\sum_{i=1}^r \sigma_i^2 + (k - r)\tilde{\sigma}^2}.
	\end{align*}
\end{thm}
In particular, the duality of the operator and trace norms arises in the $k = 1$ case of Theorem~\ref{thm:2k_dual}.

We now present some general properties of dual norms that will be of use to us. For further properties of dual norms, the interested reader is directed to any number of other sources, including \cite{Bha97,BV04,HJ85}. Let $\mnorm{\cdot}_1$ and $\mnorm{\cdot}_2$ be two norms on $\cl{H}$. Then:
\begin{align}
	\label{eq:dual_ineq}\mnorm{\cdot}_1 \leq \mnorm{\cdot}_2 \ \ \ \Longleftrightarrow & \ \ \ \mnorm{\cdot}_2^\circ \leq \mnorm{\cdot}_1^\circ, \\
	\label{eq:dual_isom}\bmnorm{A\ket{v}}_1 = \bmnorm{\ket{v}}_1 \ \ \forall \, \ket{v} \in \cl{H} \ \ \ \Longleftrightarrow & \ \ \ \bmnorm{A^\dagger\ket{v}}_1^\circ = \bmnorm{\ket{v}}_1^\circ \ \ \forall \, \ket{v} \in \cl{H}.
\end{align}

The following result allows us to rephrase dual norms, which so far we have written as the supremum~\eqref{eq:dual_defn}, as an infimum. We expect that this result is known, though we have not been able to find a reference for it.
\begin{thm}\label{thm:dual_charac_general}
	Let $\cl{S} \subseteq \cl{H}$ be a bounded set satisfying ${\rm span}(\cl{S}) = \cl{H}$ and define a norm $\mnorm{\cdot}$ by
	\begin{align*}
		\mnorm{\mathbf{v}} := \sup_{\mathbf{w} \in \cl{S}}\Big\{ \big| \langle \mathbf{v}, \mathbf{w} \rangle \big| \Big\}.
	\end{align*}
	Then $\mnorm{\cdot}^\circ$ is given by
	\begin{align*}
		\mnorm{\mathbf{v}}^\circ = \inf\Big\{ \sum_i |c_i| : \mathbf{v} = \sum_i c_i \mathbf{v}_i, \text{where } c_i \in \bb{F}, \mathbf{v}_i \in \cl{S} \ \forall \, i \Big\},
	\end{align*}
	where the infimum is taken over all such decompositions of $\mathbf{v}$.
\end{thm}

Before proving the result, we make three observations. Firstly, the conditions placed on $\cl{S}$ by Theorem~\ref{thm:dual_charac_general} are both necessary and sufficient for the quantity $\mnorm{\cdot}$ to be a norm: boundedness of $\cl{S}$ ensures that the supremum is finite, and ${\rm span}(\cl{S}) = \cl{H}$ is equivalent to the statement that $\mnorm{\mathbf{v}} = 0$ if and only if $\mathbf{v} = 0$. Secondly, every norm on $\cl{H}$ can be written in this form: we can always choose $\cl{S}$ to be the unit ball of the dual norm $\mnorm{\cdot}^\circ$. However, there are times when other choices of $\cl{S}$ are more useful. Finally, it is an elementary exercise to show that if $\cl{S}$ is closed then both the supremum and infimum in the statement of the theorem are attained (and thus can be written as a maximum and minimum, respectively).

\begin{proof}[Proof of Theorem~\ref{thm:dual_charac_general}]
	Begin by noting that if $\mathbf{w} \in \cl{S}$ and $\mnorm{\mathbf{v}} \leq 1$ then $| \langle \mathbf{v}, \mathbf{w} \rangle | \leq 1$. It follows that $\mnorm{\mathbf{w}}^{\circ} \leq 1$ whenever $\mathbf{w} \in \cl{S}$. In fact, we now show that $\mnorm{\cdot}^\circ$ is the largest norm on $\cl{H}$ with this property. To this end, let $\mnorm{\cdot}_2$ be another norm satisfying $\mnorm{\mathbf{w}}_{2}^{\circ} \leq 1$ whenever $\mathbf{w} \in \cl{S}$. Then
\begin{align*}
	\mnorm{\mathbf{v}} = \sup_{\mathbf{w} \in \cl{S}} \Big\{ \big| \langle \mathbf{v}, \mathbf{w} \rangle \big| \Big\} \leq \sup_{\mathbf{w}} \Big\{ \big| \langle \mathbf{v}, \mathbf{w} \rangle \big| : \mnorm{\mathbf{w}}_{2}^{\circ} \leq 1 \Big\} = \mnorm{\mathbf{v}}_2.
\end{align*}
Thus $\mnorm{\cdot} \leq \mnorm{\cdot}_2$, so by taking duals and using Property~\eqref{eq:dual_ineq} we see that $\mnorm{\cdot}^\circ \geq \mnorm{\cdot}_2^\circ$, as desired.

For the remainder of the proof, we denote the infimum in the statement of the theorem by $\|\cdot\|_{\textup{inf}}$. Our goal now is to show that: (a) $\|\cdot\|_{\textup{inf}}$ is a norm, (b) $\|\cdot\|_{\textup{inf}}$ satisfies $\|\mathbf{w}\|_{\textup{inf}} \leq 1$ whenever $\mathbf{w} \in \cl{S}$, and (c) $\|\cdot\|_{\textup{inf}}$ is the largest norm satisfying property~(b). The fact that $\|\cdot\|_{\textup{inf}} = \mnorm{\cdot}^\circ$ will then follow from the fact that $\mnorm{\cdot}^\circ$ is also the largest norm satisfying property~(b).

To see (a) (i.e., to prove that $\|\cdot\|_{\textup{inf}}$ is a norm), we only prove the triangle inequality, since the other properties are trivial. Fix $\varepsilon > 0$ and let $\mathbf{v} = \sum_i c_i \mathbf{v}_i$, $\mathbf{w} = \sum_i d_i \mathbf{w}_i$ be decompositions of $\mathbf{v}, \mathbf{w}$ with $\mathbf{v}_i, \mathbf{w}_i \in \cl{S}$ for all $i$, satisfying $\sum_i |c_i| \leq \|\mathbf{v}\|_{\textup{inf}} + \varepsilon$ and $\sum_i |d_i| \leq \|\mathbf{w}\|_{\textup{inf}} + \varepsilon$. Then
\begin{align*}
	\|\mathbf{v} + \mathbf{w}\|_{\textup{inf}} \leq \sum_i |c_i| + \sum_i |d_i| \leq \|\mathbf{v}\|_{\textup{inf}} + \|\mathbf{w}\|_{\textup{inf}} + 2\varepsilon.
\end{align*}
Since $\varepsilon > 0$ is arbitrary, the triangle inequality follows, so $\|\cdot\|_{\textup{inf}}$ is a norm.

To see (b) (i.e., to prove that $\|\mathbf{v}\|_{\textup{inf}} \leq 1$ whenever $\mathbf{v} \in \cl{S}$), we simply write $\mathbf{v}$ in its trivial decomposition $\mathbf{v} = \mathbf{v}$, which gives $\|\mathbf{v}\|_{\textup{inf}} \leq \sum_i c_i = c_1 = 1$.

To see (c) (i.e., to prove that $\|\cdot\|_{\textup{inf}}$ is the largest norm on $\cl{H}$ satisfying condition (b)), begin by letting $\mnorm{\cdot}_2$ be any norm on $\cl{H}$ with the property that $\mnorm{\mathbf{v}}_{2} \leq 1$ for all $\mathbf{v} \in \cl{S}$. Then using the triangle inequality for $\mnorm{\cdot}_2$ shows that if $\mathbf{v} = \sum_i c_i \mathbf{v}_i$ is any decomposition of $\mathbf{v}$ with $\mathbf{v}_i \in \cl{S}$ for all $i$, then
\begin{align*}
	\mnorm{\mathbf{v}}_2 = \mnorm{\sum_i c_i \mathbf{v}_i}_2 \leq \sum_i |c_i| \mnorm{\mathbf{v}_i}_2 = \sum_i |c_i|.
\end{align*}
Taking the infimum over all such decompositions of $\mathbf{v}$ shows that $\mnorm{\mathbf{v}}_2 \leq \|\mathbf{v}\|_{\textup{inf}}$, which completes the proof.
\end{proof}

As an example of an application of Theorem~\ref{thm:dual_charac_general}, we again consider the operator norm and trace norm on $M_n$, which we already noted are dual to each other. The theorem then says that
\begin{align*}
	\big\|X\big\| & = \inf\Big\{ \sum_i |c_i| : X = \sum_i c_i U_i \text{ with each $U_i$ unitary} \Big\}, \text{ and}\\
	\big\|X\big\|_{tr} & = \inf\Big\{ \sum_i |c_i| : X = \sum_i c_i \ketbra{w_i}{v_i} \Big\}.
\end{align*}
The above characterization of $\|\cdot\|_{tr}$ is well-known, and the infimum is attained when we write $X$ in its singular value decomposition. The characterization of $\|\cdot\|$ is perhaps slightly less well-known and interesting in its own right. Theorem~\ref{thm:dual_charac_general} also generalizes the fact that the injective and projective tensor norms are dual to each other (see \cite[Chapter~1]{DGFS08}).

\section{Basics of Quantum Entanglement}\label{sec:quantum}

Here we introduce our notation and terminology related to quantum entanglement. Our introduction to quantum information and quantum entanglement is quite brief, so the interested reader is directed to other sources such as \cite{BZ06,HHH09,NC00} for a more thorough introduction to the subject.

Throughout this work, we primarily consider three different Hilbert spaces. The first Hilbert space of interest is $\bb{C}^n$: $n$-dimensional complex Euclidean space. The second is $M_n$: the space of $n \times n$ complex matrices, equipped with the Hilbert--Schmidt inner product $\langle A, B \rangle := \Tr(A^\dagger B)$. Finally, the third Hilbert space we consider is $M_n^{H}$: the space of $n \times n$ complex Hermitian (i.e., self-adjoint) matrices, also equipped with the Hilbert--Schmidt inner product. Note that the first two Hilbert spaces are complex, while the third Hilbert space is real. We also consider tensor products of these Hilbert spaces with their natural inner products.

A pure quantum state is represented by a unit vector $\ket{v} \in \bb{C}^n$. A pure state $\ket{v} \in \bb{C}^m \otimes \bb{C}^n$ is called \emph{separable} if it can be written in the form $\ket{v} = \ket{a} \otimes \ket{b}$ for some $\ket{a} \in \bb{C}^m$, $\ket{b} \in \bb{C}^n$, and it is called \emph{entangled} otherwise. The \emph{Schmidt rank} of a pure state $\ket{v}$, which we denote by $SR(\ket{v})$, is the least integer $k$ so that we can write $\ket{v} = \sum_{i=1}^k c_i \ket{v_i}$ with each $\ket{v_i}$ separable. It is the case that $1 \leq SR(\ket{v}) \leq \min\{m,n\}$ for all $\ket{v} \in \bb{C}^m \otimes \bb{C}^n$ and $SR(\ket{v}) = 1$ if and only if $\ket{v}$ is separable. Furthermore, for every quantum state we can find orthonormal sets $\{\ket{a_i}\}\subset \bb{C}^m$ and $\{\ket{b_i}\}\subset \bb{C}^n$ and real positive coefficients (known as \emph{Schmidt coefficients}) $\alpha_1 \geq \alpha_2 \geq \cdots \geq 0$ such that $\ket{v} = \sum_{i=1}^{SR(\ket{v})}\alpha_i \ket{a_i} \otimes \ket{b_i}$ (the Schmidt rank and coefficients can be calculated for a given vector by a simple application of the singular value decomposition).

While pure states are generally rather easy to work with mathematically, not all quantum states are pure. General (i.e., potentially mixed) quantum states are represented by \emph{density matrices}: positive semidefinite matrices $\rho \in M_n^H$ satisfying $\Tr(\rho) = 1$. If $\ket{v}$ represents a pure state then the projection onto its span, $\ketbra{v}{v}$, is its density matrix representation. A general density matrix $\rho$ can be written as a convex combination of pure states: $\rho = \sum_i p_i \ketbra{v_i}{v_i}$ with $\sum_i p_i = 1$ and $p_i \geq 0$ for all $i$. If $\rho$ can be written in this way as a convex combination of separable pure states $\ket{v_i}$ then we say that $\rho$ is \emph{separable} \cite{W89}. More generally, the \emph{Schmidt number} of $\rho$, denoted $SN(\rho)$, is the least integer $k$ such that $\rho$ can be written as a convex combination of pure states $\ket{v_i}$ each with $SR(\ket{v_i}) \leq k$ \cite{TH00}. If $SN(\rho) \geq 2$ then $\rho$ is called \emph{entangled}.

An operator $Y \in (M_m \otimes M_n)^H$ is called \emph{$k$-block positive} if $\bra{v}Y\ket{v} \geq 0$ whenever $SR(\ket{v}) \leq k$. The sets of $k$-block positive operators and states with Schmidt number at most $k$ are dual to each other in the sense that $SN(\rho) \leq k$ if and only if $\Tr(\rho Y) \geq 0$ for all $k$-block positive $Y$ \cite{SSZ09}.

Fix $1 \leq k \leq \min\{m,n\}$. Four norms based on Schmidt rank and Schmidt number, which are the focus of the remainder of this paper, are as follows. In all cases, $X \in M_m \otimes M_n$ and $Y \in (M_m \otimes M_n)^H$.
\begin{align}
	\label{eq:k_rad}r^{\otimes}_k(Y) & := \sup\Big\{ \big| \bra{v}Y\ket{v} \big| : SR(\ket{v}) \leq k \Big\}, \\
	\label{eq:sk_norm}\big\|X\big\|_{S(k)} & := \sup\Big\{ \big| \bra{v}X\ket{w} \big| : SR(\ket{v}), SR(\ket{w}) \leq k \Big\}, \\
	\label{eq:sk_dual}\big\|X\big\|_{\gamma,k} & := \inf\Big\{ \sum_i |c_i| : X = \sum_i c_i\ketbra{v_i}{w_i} \text{ with } SR(\ket{v_i}),SR(\ket{w_i}) \leq k \ \forall \, i \Big\}, \text{ and} \\
	\label{eq:k_robust}R_k(Y) & := \inf\big\{ c_1 + c_2 : Y = c_1\rho_1 - c_2\rho_2 \text{ with } c_1,c_2\geq 0, SN(\rho_1),SN(\rho_2) \leq k \big\},
\end{align}
where the suprema~\eqref{eq:k_rad} and~\eqref{eq:sk_norm} are taken over all $\ket{v},\ket{w}$ and the infima~\eqref{eq:sk_dual} and~\eqref{eq:k_robust} are taken over all decompositions of the indicated form.

The norms~\eqref{eq:k_rad} and~\eqref{eq:sk_norm} can be thought of as ``$k$-local'' versions of the operator norm, and similarly the norms~\eqref{eq:sk_dual} and~\eqref{eq:k_robust} are analogous to the trace norm. In particular, in the $k = \min\{m,n\}$ case we have
\begin{align*}
	\big\|X\big\|_{S(\min\{m,n\})} = \big\|X\big\| \ \ \text{and} \ \ \big\|X\big\|_{\gamma,\min\{m,n\}} = \big\|X\big\|_{tr}.
\end{align*}
We similarly have $r^{\otimes}_{\min\{m,n\}}(Y) = \|Y\|$ and $R_{\min\{m,n\}}(Y) = \|Y\|_{tr}$ in the case when $Y$ is Hermitian. Even though the norm $r^{\otimes}_k$ can easily be defined on all of $M_m \otimes M_n$, it is more natural for us to restrict it to Hermitian operators. Furthermore, $R_k$ is a norm only on $(M_m \otimes M_n)^H$, since $c_1\rho_1 - c_2\rho_2$ is always Hermitian. The fact that every Hermitian operator $Y$ can be written in this form follows from noting that in a sum of Hermitian elementary tensors, each Hermitian matrix can be written as the difference of positive and negative parts, and terms can be regrouped to to write $Y = P - N$, where each of $P$ and $N$ is separable.

The norm~\eqref{eq:k_rad} was introduced and studied in \cite{GPMSZ10,PGMSCZ11} in the $k = 1$ case, where it was called the \emph{product numerical radius}. The norm~\eqref{eq:sk_norm} was introduced in \cite{JK10,JK11} and we have $\big\|X\big\|_{S(k)} = r^{\otimes}_{k}(X)$ when $X$ is positive semidefinite \cite[Proposition~4.5]{JK10}. The norm~\eqref{eq:sk_dual}, in the $k = 1$ case, was studied in relation to quantum entanglement in \cite{R00,Rud05} and is called the \emph{projective tensor norm}. Observe in this case that it can be written in the following slightly simpler form:
\begin{align*}
	\big\|X\big\|_{\gamma,1} & = \inf\Big\{ \sum_i |c_i| : X = \sum_i c_i\ketbra{v_i}{w_i} \otimes \ketbra{x_i}{y_i} \Big\} \\
	& = \inf\Big\{ \sum_i \big\|A_i\big\|_{tr} \big\|B_i\big\|_{tr} : Y = \sum_i A_i \otimes B_i \Big\}.
\end{align*}
Finally, the norm~\eqref{eq:k_robust} was also studied in the $k = 1$ case in \cite{Rud05}. It is easily-verified that for density matrices we have $R_k(\rho) = 2E_{R,k}(\rho) + 1$, where
\begin{align*}
	E_{R,k}(\rho) := \inf\Big\{ s : SN(\rho + s\sigma) \leq k, SN(\sigma) \leq k \Big\}.
\end{align*}
In the $k = 1$ case, $E_{R,1}$ is called the \emph{robustness of entanglement} \cite{VT99}, which represents the least amount of separable noise that can be added to a state to destroy its entanglement.

The norms~\eqref{eq:k_rad} and~\eqref{eq:sk_norm} are in some sense the natural norms to use when dealing with $k$-block positivity, when using the Hilbert space $(M_m \otimes M_n)^H$ or $M_m \otimes M_n$ respectively, as motivated by the following proposition.
\begin{prop}\label{prop:k_block_pos}
	Let $Y \in (M_m \otimes M_n)^H$. If we write $Y = cI - X$ with $X$ positive semidefinite, then the following are equivalent:
	\begin{enumerate}[(a)]
		\item $Y$ is $k$-block positive;
		\item $c \geq \|X\|_{S(k)}$; and
		\item $c \geq r^{\otimes}_k(X)$.
	\end{enumerate}
\end{prop}
\begin{proof}
	The proof is trivial and thus omitted -- see \cite[Corollary~4.9]{JK10}.
\end{proof}

Since the set of $k$-block positive operators is dual to the set of states with Schmidt number no larger than $k$, we might na\"{i}vely expect that the dual norms of $r^{\otimes}_k$ and $\|\cdot\|_{S(k)}$ similarly characterize Schmidt number. We show in the next section that their dual norms are $R_k$ and $\|\cdot\|_{\gamma,k}$, respectively, and that these norms do indeed characterize Schmidt number.

\section{Duality of Schmidt Rank Norms}\label{sec:dual_ent_norms}

We begin by showing that the norms~\eqref{eq:k_rad} and~\eqref{eq:k_robust} are dual to each other, and that the norms~\eqref{eq:sk_norm} and~\eqref{eq:sk_dual} are dual to each other.
\begin{thm}\label{thm:dual_norm_sk}
	Let $X \in M_m \otimes M_n$ and $Y \in (M_m \otimes M_n)^H$. Then
	\begin{align*}
		\big\| X \big\|_{S(k)}^{\circ} & = \big\|X\big\|_{\gamma,k} \ \ \text{ and } \ \ r^{\otimes}_k(Y)^\circ = R_k(Y).
	\end{align*}
\end{thm}
\begin{proof}
	To see the first equality, simply use Theorem~\ref{thm:dual_charac_general} with $\cl{H} = M_m \otimes M_n$ and $\cl{S} = \big\{ \ketbra{v}{w} : SR(\ket{v}),SR(\ket{w}) \leq k \big\}$. For the second equality, similarly let $\cl{H} = (M_m \otimes M_n)^H$ and $\cl{S} = \big\{ \ketbra{v}{v} : SR(\ket{v}) \leq k \big\}$ to see that
	\begin{align*}
		r^{\otimes}_k(Y)^\circ & = \inf\Big\{ \sum_i |c_i| : Y = \sum_i c_i\ketbra{v_i}{v_i} \text{ with } c_i \in \bb{R} \text{ and } SR(\ket{v_i}) \leq k \ \forall \, i \Big\}.
	\end{align*}
	By simply grouping the positive coefficients $\{c_i\}$ together, and similarly grouping the negative coefficients together, we see that
	\begin{align*}
		r^{\otimes}_k(Y)^\circ & = \inf\big\{ c_1 + c_2 : Y = c_1\rho_1 - c_2\rho_2 \text{ with } c_1,c_2\geq 0, SN(\rho_1),SN(\rho_2) \leq k \big\},
	\end{align*}
	as desired.
\end{proof}
A completely different proof that $\| X \|_{S(k)}^{\circ} = \|X\|_{\gamma,k}$, based on minimal and maximal operator spaces, was given in \cite{Joh12}. Indeed, the $S(k)$-norm is the $k$-minimal $L^\infty$-matrix norm on $M_n$ \cite{JKPP11}, so the dual of the $S(k)$-norm is analogously the $k$-maximal $L^1$-matrix norm on $M_n$.

Similarly, it was shown in \cite[Theorem~9]{JKPP11} that $r^{\otimes}_k(\cdot)$ is the natural norm on the $k$-super minimal operator system on $M_n$, introduced in \cite{Xthesis,Xha11}. This observation leads immediately to the following alternate characterizations of $r^{\otimes}_k(\cdot)$.
\begin{thm}\label{thm:op_sys_char}
	Let $X \in (M_m \otimes M_n)^H$. Then
	\begin{align}
		\label{eq:prod_rad_first}r^{\otimes}_k(X) & = \inf \big\{ s : sI \pm X \text{ are both $k$-block positive} \big\} \\
		\label{eq:prod_rad_second}& = \inf \left\{ s : \begin{bmatrix}sI_m \otimes I_n & X \\ X^\dagger & sI_m \otimes I_n\end{bmatrix} \in M_{2m} \otimes M_n \text{ is $k$-block positive} \right\}.
	\end{align}
\end{thm}
\begin{proof}
	As already mentioned, $r^{\otimes}_k(\cdot)$ is the natural norm on Hermitian elements on the $k$-super minimal operator system on $M_n$. Various norms on operator systems were studied in \cite{PT09} -- in their notation, we have $r^{\otimes}_k(\cdot) = \|\cdot\|_m$, the minimal extension of the operator system norm from Hermitian elements to all of $M_m \otimes M_n$. Similarly, the norm~\eqref{eq:prod_rad_first} is the ``order norm'' $\|\cdot\|_{or}$ and the norm~\eqref{eq:prod_rad_second} is the natural operator system norm. Since all of these norms coincide on Hermitian matrices, the result follows.
\end{proof}

In general, the infimum~\eqref{eq:prod_rad_second} on non-Hermitian elements is not necessarily equal to $r^{\otimes}_k(\cdot)$, but rather is an upper bound of it.

We now begin presenting consequences of the duality provided by Theorem~\ref{thm:dual_norm_sk}. Our first result in this direction generalizes the fact that a density matrix $\rho$ is separable if and only if $\|\rho\|_{\gamma,1} = 1$ \cite{R00}, which is known as the \emph{cross norm criterion} for separability.
\begin{thm}\label{thm:dual_sep_equiv}
	Let $\rho \in M_m \otimes M_n$ be a density matrix. Then the following are equivalent:
	\begin{enumerate}[(a)]
		\item $SN(\rho) \leq k$;
		\item $\|\rho\|_{\gamma,k} = 1$; and
		\item $R_k(\rho) = 1$.
	\end{enumerate}
\end{thm}
\begin{proof}
	Note that $R_k(\rho) \geq \|\rho\|_{\gamma,k} \geq \|\rho\|_{tr} = 1$ for all $\rho$, so we only need to show two implications:
	\begin{enumerate}[(i)]
		\item if $SN(\rho) \leq k$ then $R_k(\rho) \leq 1$, and
		\item if $\|\rho\|_{\gamma,k} \leq 1$ then $SN(\rho) \leq k$.
	\end{enumerate}
	The implication (i) follows from the easily-verified facts that $R_k(\rho) = 2E_{R,k}(\rho) + 1$ for all $\rho$ and $E_{R,k}(\rho) = 0$ if and only if $SN(\rho) \leq k$.
	
	To see the implication (ii), suppose $\|\rho\|_{\gamma,k} \leq 1$ and let $Y \in (M_m \otimes M_n)^H$ be $k$-block positive. If we write $Y = cI - X$ with $X$ positive semidefinite then $c \geq \|X\|_{S(k)}$, by Proposition~\ref{prop:k_block_pos}. We then have
	\begin{align*}
		\Tr(\rho Y) = \Tr\big(\rho(cI - X)\big) = c - \Tr(\rho X) \geq c - \big\|X\big\|_{S(k)} \geq 0,
	\end{align*}
	where we used the duality of Theorem~\ref{thm:dual_norm_sk} in the second-last inequality. Since $Y$ is an arbitrary $k$-block positive operator, it follows that $SN(\rho) \leq k$, which completes the proof.
\end{proof}

\section{Values on Pure States}\label{sec:pure_states}

We now consider the problem of computing the norms~\eqref{eq:k_rad}, \eqref{eq:sk_norm}, \eqref{eq:sk_dual}, and \eqref{eq:k_robust} on pure states $\ketbra{v}{v}$. Because each of these norms is invariant under operations of the form $X \mapsto (U \otimes V)X(U \otimes V)^\dagger$, where $U \in M_m$ and $V \in M_n$ are unitary operators, we know that their values on pure states depend only on the state's Schmidt coefficients $\alpha_1 \geq \alpha_2 \geq \cdots \geq \alpha_{\min\{m,n\}} \geq 0$. Thus, we look for formulas for these norms on pure states in terms of Schmidt coefficients.

It was shown in \cite[Theorem~3.3 and Proposition~4.3]{JK10} that
\begin{align*}
	r^{\otimes}_k(\ketbra{v}{v}) = \big\|\ketbra{v}{v}\big\|_{S(k)} = \sum_{i=1}^k \alpha_i^2.
\end{align*}

We thus move directly to the problem of calculating $\|\ketbra{v}{v}\|_{\gamma,k}$. It was shown in \cite{Rud01} that in the $k = 1$ case we have
\begin{align}\label{eq:gamma_norm_1_pure}
	\big\|\ketbra{v}{v}\big\|_{\gamma,1} = \left(\sum_{i=1}^{\min\{m,n\}} \alpha_i\right)^2.
\end{align}
The following result establishes the natural generalization of this fact for arbitrary $k$.
\begin{thm}\label{thm:pure_state_sk_dual}
	Let $\ket{v} \in \bb{C}^m \otimes \bb{C}^n$ and fix $1 \leq k \leq \min\{m,n\}$. Let $r$ be the largest index $1 \leq r < k$ such that $\alpha_r > \sum_{i=r+1}^{\min\{m,n\}}\alpha_{i}/(k-r)$ (or take $r = 0$ if no such index exists). Also define $\tilde{\alpha} := \sum_{i=r+1}^{\min\{m,n\}}\alpha_i/(k-r)$. Then	
	\begin{align}\label{eq:sk_dual_pure}
		\big\|\ketbra{v}{v}\big\|_{\gamma,k} = \sum_{i=1}^r \alpha_i^2 + (k - r)\tilde{\alpha}^2.
	\end{align}
\end{thm}
\begin{proof}
	To see that $\big\|\ketbra{v}{v}\big\|_{\gamma,k} \geq \sum_{i=1}^r \alpha_i^2 + (k - r)\tilde{\alpha}^2$, use the duality of Theorem~\ref{thm:dual_norm_sk} to see that
	\begin{align*}
		\big\|\ketbra{v}{v}\big\|_{\gamma,k} & = \sup \Big\{ \big| \Tr(\ketbra{v}{v}X) \big| : \big\|X\big\|_{S(k)} \leq 1 \Big\} \geq \sup \Big\{ \big| \bra{v}\mathbf{w} \big|^2 : \big\|\mathbf{w}\mathbf{w}^\dagger\big\|_{S(k)} \leq 1 \Big\}.
	\end{align*}
	If we now define the norm
	\begin{align*}
		\big\|\mathbf{w}\big\|_{s(k)} := \sup \Big\{ \big| \bra{v}\mathbf{w} \big| : SR(\ket{v}) \leq k \Big\},
	\end{align*}
	then it is clear that $\|\mathbf{w}\mathbf{w}^\dagger\|_{S(k)} \leq 1$ if and only if $\|\mathbf{w}\|_{s(k)} \leq 1$. Thus
	\begin{align*}
	\big\|\ketbra{v}{v}\big\|_{\gamma,k} & \geq \sup \Big\{ \big| \bra{v}\mathbf{w} \big|^2 : \big\|\mathbf{w}\big\|_{s(k)} \leq 1 \Big\} = \left(\big\|\ket{v}\big\|_{s(k)}^\circ\right)^2.
	\end{align*}
	
	By using the fact that $\|\ket{v}\|_{s(k)} = \sqrt{\sum_{i=1}^k \alpha_i^2}$ (see \cite[Theorem~3.3]{JK10}) and the duality result \cite[Theorem~3.3]{MF85}, we see that
	\begin{align}\label{eq:sk_vec_dual}
		\big\|\ket{v}\big\|_{s(k)}^\circ = \sqrt{\sum_{i=1}^r \alpha_i^2 + (k - r)\tilde{\alpha}^2},
	\end{align}
	where $r$ and $\tilde{\alpha}$ are as in the statement of the theorem (for a more explicit proof of Equation~\eqref{eq:sk_vec_dual}, see \cite[Section~4.1.2]{Joh12}). This completes the ``$\geq$'' direction of the proof.
	
	To see the opposite inequality, use Theorem~\ref{thm:dual_charac_general} to see that
	\begin{align*}
		\big\|\ket{v}\big\|_{s(k)}^\circ = \inf\Big\{ \sum_i |c_i| : \ket{v} = \sum_i c_i\ket{v_i} \text{ with } SR(\ket{v_i}) \leq k \ \forall \, i \Big\},
	\end{align*}
	where the infimum is taken over all decompositions of $\ket{v}$ of the given form. It follows that
	\begin{align*}
		\big\|\ketbra{v}{v}\big\|_{\gamma,k} & = \inf\Big\{ \sum_i |c_i| : \ketbra{v}{v} = \sum_i c_i\ketbra{v_i}{w_i} \text{ with } SR(\ket{v_i}),SR(\ket{w_i}) \leq k \ \forall \, i \Big\} \\
		& \leq \inf\Big\{ \left(\sum_{i} |c_i|\right)^2 : \ket{v} = \sum_i c_i\ket{v_i} \text{ with } SR(\ket{v_i}) \leq k \ \forall \, i \Big\} \\
		& = \left(\big\|\ket{v}\big\|_{s(k)}^\circ\right)^2.
	\end{align*}
	By using Equation~\eqref{eq:sk_vec_dual} again, the desired inequality follows, and the proof is complete.
\end{proof}

Before proceeding, we make some observations about Theorem~\ref{thm:pure_state_sk_dual}. If $k = 1$ then the only possible choice for $r$ is $r = 0$, so the norm reduces to simply Equation~\eqref{eq:gamma_norm_1_pure} in this case, as it should. At the other extreme, if $k = \min\{m,n\}$ then $r = \min\{m,n\} - 1$. Thus
\begin{align*}
	\big\|\ketbra{v}{v}\big\|_{\gamma,\min\{m,n\}} = \sum_{i=1}^{\min\{m,n\}} \alpha_i^2 = 1,
\end{align*}
which is just the trace norm of $\ketbra{v}{v}$, as expected. Similarly, if $1 \leq k \leq \min\{m,n\}$ and $SR(\ket{v}) \leq k$ then this same argument shows that $\|\ketbra{v}{v}\|_{\gamma,k} = 1$, which we expect from Theorem~\ref{thm:dual_sep_equiv}.

Note that Theorem~\ref{thm:pure_state_sk_dual} generalizes in the obvious way to non-Hermitian rank-one operators of the form $\ketbra{v}{w}$. Indeed, if $\{\alpha_i\}$ and $\{\beta_i\}$ are the Schmidt coefficients of $\ket{v}$ and $\ket{w}$ respectively, we define $r$ as in Theorem~\ref{thm:pure_state_sk_dual} (and analogously define $s$ for $\{\beta_i\}$), and we let $\tilde{\alpha}$ (and analogously $\tilde{\beta}$) be as in the theorem, then
\begin{align*}
	\big\|\ketbra{v}{w}\big\|_{\gamma,k} = \sqrt{\sum_{i=1}^r \alpha_i^2 + (k - r)\tilde{\alpha}^2}\sqrt{\sum_{i=1}^s \beta_i^2 + (k - s)\tilde{\beta}^2}.
\end{align*}
This provides the natural generalization of \cite[Proposition~11]{Rud05}.

Finally, we wish to obtain a formula for $R_k(\ketbra{v}{v})$. We conjecture (but do not prove) that
\begin{align*}
	R_k(\ketbra{v}{v}) = 2\big\|\ketbra{v}{v}\big\|_{\gamma,k} - 1.
\end{align*}
Indeed, this formula was proved in the $k = 1$ case in \cite{VT99} and holds trivially in the $k = \min\{m,n\}$ case since the left hand side and right hand side both equal $1$. We are not aware of a proof or a counter-example for the intermediate values of $k$.

\section{Isometry Groups}\label{sec:isometry}

The entanglement norms we are considering are all invariant under local unitaries -- indeed, this is typically included as an axiom for what makes a ``good'' entanglement measure \cite{Vid00}. Slightly more generally, we consider unitary matrices $U \in M_m \otimes M_n$ of the form
\begin{align}\label{eq:localU}
	U = U_1 \otimes U_2	\ \ \text{ or } \ \ n = m \text{ and } U = S(U_1 \otimes U_2),
\end{align}
where $U_1 \in M_m$ and $U_2 \in M_n$ are unitary matrices and $S \in M_n \otimes M_n$ is the \emph{swap operator} defined on elementary tensors by $S(\ket{a}\otimes\ket{b}) = \ket{b}\otimes\ket{a}$.

It is easily-verified that, for all $k$, if $U$ and $V$ are unitary matrices of the form~\eqref{eq:localU}, then
\begin{align*}
	\big\| UXV \big\|_{S(k)} = \big\| X \big\|_{S(k)} \ \ & \text{ and } \ \ \big\| UXV \big\|_{\gamma,k} = \big\| X \big\|_{\gamma,k} \quad \forall \, X \in M_m \otimes M_n, \ \text{ and} \\
	r^{\otimes}_k(UXU^\dagger) = r^{\otimes}_k(X) \ \ & \text{ and } \ \ R_k(UXU^\dagger) = R_k(X) \quad \forall \, X \in (M_m \otimes M_n)^H.
\end{align*}
Using Theorem~\ref{thm:dual_norm_sk}, we can now answer the question of what other linear maps preserve these norms (i.e., we derive the structure of the isometry groups of these norms). In all cases, we see that the local unitaries are almost the only preservers of these norms.
\begin{thm}\label{thm:sk_iso}
	Let $1 \leq k < \min\{m,n\}$ and let $\Phi : M_m \otimes M_n \rightarrow M_m \otimes M_n$ be linear. The following are equivalent:
	\begin{enumerate}[(a)]
		\item $\big\|\Phi(X)\big\|_{S(k)} = \big\|X\big\|_{S(k)}$ for all $X \in M_m \otimes M_n$;
		\item $\big\|\Phi(X)\big\|_{\gamma,k} = \big\|X\big\|_{\gamma,k}$ for all $X \in M_m \otimes M_n$; and
		\item $\Phi$ can be written as a composition of one or more of the following maps:
		\begin{itemize}
			\item $X \mapsto UXV$, where $U$ and $V$ are unitary matrices of the form~\eqref{eq:localU},
			\item the transpose map $T$, and
			\item if $k = 1$, the partial transpose map $(id_m \otimes T)$.
		\end{itemize}
	\end{enumerate}
\end{thm}
\begin{proof}
	The equivalence of (a) and (c) was proved in \cite{Joh11}. To see that (a) and (b) are equivalent, simply recall Property~\eqref{eq:dual_isom}, which says that if $\mnorm{\cdot}$ is any norm on $M_m \otimes M_n$ then $\mnorm{\Phi(X)} = \mnorm{X}$ for all $X \in M_m \otimes M_n$ if and only if $\mnorm{\Phi^\dagger(X)}^\circ = \mnorm{X}^\circ$ for all $X \in M_m \otimes M_n$, where $\Phi^\dagger$ is the adjoint map in the Hilbert--Schmidt inner product defined by $\Tr(A^\dagger\Phi^\dagger(B)) = \Tr(\Phi(A)^\dagger B)$ for all $A,B \in M_m \otimes M_n$.
	
	Now simply note that $T^\dagger = T$, $(id_m \otimes T)^\dagger = (id_m \otimes T)$, and if $\Phi(X) = UXV$ then $\Phi^\dagger(X) = U^\dagger XV^\dagger$, and each of $U^\dagger$ and $V^\dagger$ are of the form~\eqref{eq:localU} whenever $U$ and $V$ have that form. The result then follows from Theorem~\ref{thm:dual_norm_sk}.
\end{proof}

Using these techniques we could similarly derive the isometry group of one of $r^{\otimes}_k(\cdot)$ or $R_k(\cdot)$ from the isometry group of the other one. However, to our knowledge the isometry group has not yet been derived for either of these norms.

\section{Realignment Criterion for Arbitrary Schmidt Number}\label{sec:realignment}

Since the partial transpose map $id \otimes T$ and multiplication on the right by the swap operator $S$ both preserve the norm $\|\cdot\|_{\gamma,1}$ (see Theorem~\ref{thm:sk_iso}), it follows that the \emph{realignment map} $L : M_{m,n} \otimes M_{r,s} \rightarrow M_{m,r} \otimes M_{n,s}$ defined by $L(X) = (id \otimes T)(XS)S$ also satisfies $\|L(X)\|_{\gamma,1} = \|X\|_{\gamma,1}$ for all $X$. An immediate but important consequence of this observation is the fact that if $\rho$ is separable then $\|L(\rho)\|_{tr} \leq \|L(\rho)\|_{\gamma,1} = \|\rho\|_{\gamma,1} = 1$, where the final equality comes from Theorem~\ref{thm:dual_sep_equiv}.

The fact that $\|L(\rho)\|_{tr} \leq 1$ whenever $\rho$ is separable is known as the \emph{realignment criterion} \cite{CW03} or the \emph{computable cross norm criterion} \cite{R03}. We now present a natural generalization of this criterion for arbitrary Schmidt number, which uses the norms of Theorem~\ref{thm:2k_dual} rather than the trace norm.
\begin{thm}\label{thm:gen_realign}
	If $\rho \in M_m \otimes M_n$ has $SN(\rho) \leq k$ then $\|L(\rho)\|_{(k^2,2)}^{\circ} \leq 1$.
\end{thm}
\begin{proof}
	Suppose $SN(\rho) \leq k$ and begin by writing $\rho$ as a convex combination of projections onto states with Schmidt rank no greater than $k$:
	\begin{align*}
		\rho = \sum_i p_i \sum_{j,\ell=1}^k \alpha_{ij}\alpha_{i\ell}\ketbra{v_{ij}}{v_{i\ell}} \otimes \ketbra{w_{ij}}{w_{i\ell}}
	\end{align*}
	Then
	\begin{align*}
		L(\rho) = \sum_i p_i \left(\sum_{j=1}^k \alpha_{ij}\ket{v_{ij}}\overline{\bra{w_{ij}}}\right) \otimes \left(\sum_{\ell=1}^k\alpha_{i\ell}\overline{\ket{v_{i\ell}}}\bra{w_{i\ell}}\right).
	\end{align*}
	If we define $A_i := \sum_{j=1}^k \alpha_{ij}\ket{v_{ij}}\overline{\bra{w_{ij}}}$ then we have $L(\rho) = \sum_i p_i A_i \otimes \overline{A_i}$, where ${\rm rank}(A_i) \leq k$ and $\|A_i\|_F = 1$ for all $i$. In particular then, we have $L(\rho) = \sum_i p_i B_i$, where ${\rm rank}(B_i) \leq k^2$ and $\|B_i\|_F = 1$ for all $i$. Let $\mnorm{\cdot}$ be a norm with the property that $\mnorm{X} = \|X\|_F$ for all $X$ with ${\rm rank}(X) \leq k^2$. Then
\begin{align*}
	\bmnorm{L(\rho)} = \mnorm{\sum_i p_i B_i} \leq \sum_i p_i \mnorm{B_i} = \sum_i p_i \big\|B_i\big\|_F = \sum_i p_i = 1.
\end{align*}
All that remains is to make a suitable choice for $\mnorm{\cdot}$, so that this test for Schmidt number is as strong as possible. To this end, notice that $\|\cdot\|_{(k^2,2)}$ is clearly the smallest norm with the required rank property. Also notice that, because the Frobenius norm is self-dual, $\|\cdot\|_{(k^2,2)}^{\circ}$ must satisfy the same rank property, and in particular must be the largest such matrix norm. We thus choose $\mnorm{\cdot} = \|\cdot\|_{(k^2,2)}^{\circ}$, which completes the proof.
\end{proof}

Notice that when $k = 1$, $\|\cdot\|_{(k^2,2)}^{\circ} = \|\cdot\|_{tr}$, so Theorem~\ref{thm:gen_realign} gives the standard realignment criterion in this case. On the other extreme, if $k = \min\{m,n\}$ then $\|\cdot\|_{(k^2,2)}^{\circ} = \|\cdot\|_{F}$. Because $L$ preserves the Frobenius norm, Theorem~\ref{thm:gen_realign} then simply says that $\|\rho\|_{F} \leq 1$ for all quantum states $\rho$, which is trivially true because $\|\rho\|_{F} \leq \|\rho\|_{tr} = 1$. The conditions given for the remaining values of $k$ are all non-trivial, yet easy to compute.

To help motivate the idea that Theorem~\ref{thm:gen_realign} provides the ``right'' generalization of the realignment criterion, we now note that it provides a test that is both necessary and sufficient on pure states.
\begin{thm}
	Let $\ket{v} \in \bb{C}^m \otimes \bb{C}^n$. Then $SR(\ket{v}) \leq k$ if and only if $\|L(\ketbra{v}{v})\|_{(k^2,2)}^{\circ} \leq 1$.
\end{thm}
\begin{proof}
	The ``only if'' implication is provided by Theorem~\ref{thm:gen_realign}. For the ``if'' direction, write
	\begin{align*}
		\ket{v} = \sum_{i=1}^{SR(\ket{v})} \alpha_i \ket{v_i} \otimes \ket{w_i}.
	\end{align*}
	Then
	\begin{align*}
		L(\ketbra{v}{v}) & = \left(\sum_{i=1}^{SR(\ket{v})} \alpha_{i}\ket{v_{i}}\overline{\bra{w_{i}}} \right) \otimes \left(\sum_{i=1}^{SR(\ket{v})} \alpha_{i}\overline{\ket{v_{i}}}\bra{w_{i}} \right).
	\end{align*}
	In particular, if $SR(\ket{v}) > k$ then ${\rm rank}(L(\ketbra{v}{v})) > k^2$ and $\|L(\ketbra{v}{v})\|_F = 1$. We now prove by contradiction that if $SR(\ket{v}) > k$ then $\|L(\ketbra{v}{v})\|_{(k^2,2)}^{\circ} > 1$. Begin by using Theorem~\ref{thm:dual_charac_general} with the norm $\|\cdot\|_{(k^2,2)}$ to see that
	\begin{align*}
		\big\|X\big\|_{(k^2,2)}^{\circ} & = \inf\Big\{ \sum_i |c_i| : X = \sum_i c_i Y_i \text{ with ${\rm rank}(Y_i) \leq k^2$ and $\|Y_i\|_F \leq 1$ for all $i$} \Big\}.
	\end{align*}
	Now assume that $\|L(\ketbra{v}{v})\|_{(k^2,2)}^{\circ} \leq 1$ so that (by closedness of the set $\{Y : {\rm rank}(Y) \leq k^2, \|Y\|_F \leq 1\}$) there exists a decomposition $L(\ketbra{v}{v}) = \sum_i c_i Y_i$ with $\sum_j |c_j| \leq 1$, ${\rm rank}(Y_i) \leq k^2$, and $\|Y_i\|_F \leq 1$ for all $i$. Then
	\begin{align*}
		\big\|L(\ketbra{v}{v})\big\|_F = \left\|\sum_i c_i Y_i\right\|_F \leq \sum_i |c_i| \leq 1.
	\end{align*}
	Since we already saw that $\|L(\ketbra{v}{v})\|_F = 1$, the inequalities above must actually be equality. However, the first inequality is simply the triangle inequality, and equality is attained in the triangle inequality for the Frobenius norm if and only if the span set $\{Y_i\}$ has dimension $1$ (i.e., if and only if each of the $Y_i$'s are multiples of each other). However, since ${\rm rank}(Y_i) \leq k^2$ for all $i$, we then have ${\rm rank}(L(\ketbra{v}{v})) \leq k^2$ as well, which contradicts the fact that ${\rm rank}(L(\ketbra{v}{v})) > k^2$, as we already saw. We thus conclude that $\|L(\ketbra{v}{v})\|_{(k^2,2)}^{\circ} > 1$, which completes the proof.
\end{proof}

Notice that $\|\cdot\|_{(k^2,2)} \leq k\|\cdot\|$, so $\|\cdot\|_{tr} \leq k\|\cdot\|_{(k^2,2)}^\circ$. By combining this observation with Theorem~\ref{thm:gen_realign}, we arrive at a weaker generalization of the realignment criterion that says $\|L(\rho)\|_{tr} \leq k$ whenever $SN(\rho) \leq k$.

We close by noting that Theorem~\ref{thm:gen_realign} can be strenghtened further by using the local filtering technique described in \cite{GGHE08}. In that paper, it was noted that we can apply a local filtering operation to $\rho$ that does not change its Schmidt number, yet makes it ``more entangled'' in the sense that it is more susceptible to being detected by separability criteria. In particular, \cite[Proposition~IV.13]{GGHE08} (i.e., the \emph{filter covariance matrix criterion}) is the statement that results from first applying a local filter to $\rho$ and then applying the realignment criteria. One can similarly strengthen Theorem~\ref{thm:gen_realign} by first applying a local filter to $\rho$ and then using the statement of the theorem.

%

\vspace{0.1in} \noindent{\bf Acknowledgements.} N.J. was supported by the University of Guelph Brock Scholarship, NSERC of Canada, and the Mprime Network. D.W.K. was supported by NSERC Discovery Grant 400160, NSERC Discovery Accelerator Supplement 400233, and Ontario Early Researcher Award 048142. This paper makes up part of the first author's doctoral thesis.

\bibliographystyle{alpha}
\bibliography{quantum}
\end{document}